\documentclass[11pt]{article}

\usepackage[left=1cm,right=1cm,
    top=1.5cm,bottom=1.5cm,bindingoffset=0cm]{geometry}
    \usepackage{amssymb, amsmath, amsfonts, amsthm}
        \newcommand{\R}{\mathbb R}
        \newcommand{\Z}{\mathbb Z}

        \renewcommand{\P}{\mathbb P}
         \newcommand{\V}{\mathbb {V}}
          \newcommand{\Id}{ {I}}
           \newcommand{\PGL}{\P\mathrm{GL}}
           \newcommand{\Poly}{\mathcal{P}}
 
\usepackage{hyperref}
\usepackage{bbold}
\usepackage{mathrsfs}
\usepackage{MnSymbol}
\usepackage{fullpage}
\usepackage{color}
\usepackage{graphicx}
\usepackage{caption}
\usepackage{subcaption}
\usepackage{yfonts}
\usepackage{changepage}
\usepackage{bm}
\usepackage[inline]{enumitem}
\usepackage[sort&compress, numbers]{natbib}
\usepackage{pgfplots}
\pgfplotsset{compat=1.14}
\usetikzlibrary{positioning}
\usetikzlibrary{decorations.text}
\usetikzlibrary{decorations.pathmorphing}
\usetikzlibrary{decorations.markings}
\usetikzlibrary{intersections}
\usetikzlibrary{cd}

\usepackage{mathtools}
\mathtoolsset{showonlyrefs}

\newtheorem{lemma}{Lemma}[section]
\newtheorem{proposition}[lemma]{Proposition}
\newtheorem{theorem}[lemma]{Theorem}

\theoremstyle{definition}

\newtheorem{definition}[lemma]{Definition}

\usepackage{etoolbox}
\makeatletter
\pretocmd\@maketitle
  {\def\@makefnmark{\hbox{\@textsuperscript{\normalfont\@thefnmark}}}}
  {}{\FAIL}
\makeatother

\title{The limit point of the pentagram map and infinitesimal monodromy}

\author{Quinton Aboud\thanks{Department of Mathematics, University of Arizona, e-mail: \tt{aboud@math.arizona.edu}} and Anton Izosimov\thanks{Department of Mathematics, University of Arizona, e-mail: \tt{izosimov@math.arizona.edu}}}

\date{}

\begin{document}

\maketitle
\abstract{ 

The pentagram map takes a planar polygon $P$ to a polygon $P'$ whose vertices are the intersection
points of consecutive shortest diagonals of $P$. The orbit of a convex polygon under this map is a sequence of polygons which converges exponentially to a point. Furthermore, as recently proved by Glick, coordinates of that limit point can be computed as an eigenvector of a certain operator associated with the polygon. In the present paper we show that Glick's operator can be interpreted as the \textit{infinitesimal monodromy} of the polygon. Namely, there exists a certain natural infinitesimal perturbation of a polygon, which is again a polygon but in general not closed; what Glick's operator measures is the extent to which this perturbed polygon does not close up. }


\section{Introduction}
The pentagram map, introduced by R.\,Schwartz in \cite{schwartz1992pentagram}, is a discrete dynamical system on the space of planar polygons. 
The definition of this map is illustrated in Figure \ref{Fig1}: the image of the polygon $P$ under the pentagram map is the polygon $P'$ whose vertices are the intersection points of consecutive shortest diagonals of~$P$ (i.e.,  diagonals connecting second-nearest vertices). 

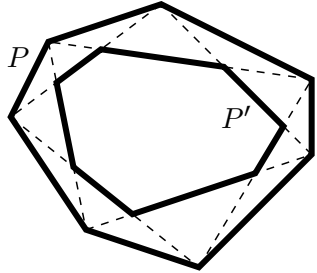
\begin{figure}[h]
\centering
\begin{tikzpicture}[thick, scale = 1]
\coordinate (VK7) at (0,0);
\coordinate (VK6) at (1.5,-0.5);
\coordinate (VK5) at (3,1);
\coordinate (VK4) at (3,2);
\coordinate (VK3) at (1,3);
\coordinate (VK2) at (-0.5,2.5);
\coordinate (VK1) at (-1,1.5);

\draw  [line width=0.8mm]  (VK7) -- (VK6) -- (VK5) -- (VK4) -- (VK3) -- (VK2) -- (VK1) -- cycle;
\draw [dashed, line width=0.2mm, name path=AC] (VK7) -- (VK5);
\draw [dashed,line width=0.2mm, name path=BD] (VK6) -- (VK4);
\draw [dashed,line width=0.2mm, name path=CE] (VK5) -- (VK3);
\draw [dashed,line width=0.2mm, name path=DF] (VK4) -- (VK2);
\draw [dashed,line width=0.2mm, name path=EG] (VK3) -- (VK1);
\draw [dashed,line width=0.2mm, name path=FA] (VK2) -- (VK7);
\draw [dashed,line width=0.2mm, name path=GB] (VK1) -- (VK6);

\path [name intersections={of=AC and BD,by=Bp}];
\path [name intersections={of=BD and CE,by=Cp}];
\path [name intersections={of=CE and DF,by=Dp}];
\path [name intersections={of=DF and EG,by=Ep}];
\path [name intersections={of=EG and FA,by=Fp}];
\path [name intersections={of=FA and GB,by=Gp}];
\path [name intersections={of=GB and AC,by=Ap}];

\draw  [line width=0.8mm]  (Ap) -- (Bp) -- (Cp) -- (Dp) -- (Ep) -- (Fp) -- (Gp) -- cycle;

\node at (-0.9,2.3) () {$P$};
\node at (2,1.5) () {$P'$};

\end{tikzpicture}
\caption{The pentagram map.}\label{Fig1}
\end{figure}

The pentagram map has been an especially popular topic in the last decade, mainly due to its connections with integrability \cite{ovsienko2010pentagram, soloviev2013integrability} and the theory of cluster algebras \cite{GLICK20111019, Gekhtman2016, fock2014loop}. Most works on the the pentagram map regard it as a dynamical system on the space of polygons modulo projective equivalence. And indeed that is the setting where most remarkable features of that map such as integrability reveal themselves. That said, the pentagram map on actual {polygons} (as opposed to projective equivalence classes) also has interesting geometry. One of the early results in this direction was Schwartz's proof of the exponential convergence of successive images of a convex polygon under the pentagram map to a point (see Figure \ref{Fig2}). That limit point is a natural invariant of a polygon and can be thought of as a projectively natural version of the center of mass. However, it is not clear a priori whether this limit point can be expressed in terms of coordinates of the vertices by any kind of an explicit formula. A remarkable recent result by M.\,Glick \cite{glick2020limit} is that this dependence is in fact algebraic. Moreover, there exists an operator in $\R^3$ whose matrix entries are rational in terms of polygon's vertices, while the coordinates of the limit point are given by an eigenvector of that operator. Therefore, coordinates of the limit point can be found by solving a cubic equation.\par

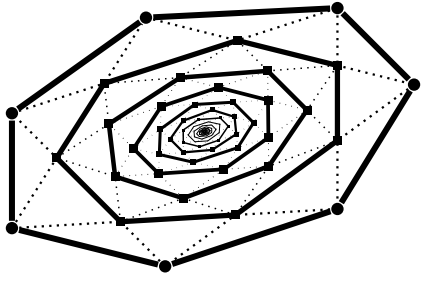
\begin{figure}[t]

\begin{center}
    \begin{tikzpicture}[thick,scale=0.15, mycirc/.style={circle,fill=black, inner sep=1.5pt},
    every node/.style={scale=1.1}]
        \node[mycirc] (A1) at (-16.97,-8.49) {};
        \node[mycirc] (B1) at (-16.97,1.70) {};
        \node[mycirc] (C1) at (-5.09,10.18) {};
        \node[mycirc] (D1) at (11.88,11.03) {};
        \node[mycirc] (E1) at (18.67,4.24) {};
        \node[mycirc] (F1) at (11.88,-6.79) {};
        \node[mycirc] (G1) at (-3.39,-11.88) {};

        \draw[line width=.8mm] (A1) -- (B1) -- (C1) -- (D1) -- (E1) -- (F1) -- (G1) -- (A1);

        \draw[name path=lineAC1,line width=.3mm, dotted] (A1) -- (C1); 
        \draw[name path=lineCE1,line width=.3mm, dotted] (C1) -- (E1);
        \draw[name path=lineEG1,line width=.3mm, dotted] (E1) -- (G1);
        \draw[name path=lineGB1,line width=.3mm, dotted] (G1) -- (B1);
        \draw[name path=lineBD1,line width=.3mm, dotted] (B1) -- (D1);
        \draw[name path=lineDF1,line width=.3mm, dotted] (D1) -- (F1);
        \draw[name path=lineFA1,line width=.3mm, dotted] (F1) -- (A1);

         \path [name intersections={of=lineGB1 and lineAC1,by=A2}];
             \node [fill=black,inner sep=1.5pt] at (A2) {};
         \path [name intersections={of=lineAC1 and lineBD1,by=B2}];
             \node [fill=black,inner sep=1.5pt] at (B2) {};
         \path [name intersections={of=lineBD1 and lineCE1,by=C2}];
             \node [fill=black,inner sep=1.5pt] at (C2) {};
         \path [name intersections={of=lineCE1 and lineDF1,by=D2}];
             \node [fill=black,inner sep=1.5pt] at (D2) {};
         \path [name intersections={of=lineDF1 and lineEG1,by=E2}];
             \node [fill=black,inner sep=1.5pt] at (E2) {};
         \path [name intersections={of=lineEG1 and lineFA1,by=F2}];
             \node [fill=black,inner sep=1.5pt] at (F2) {};
         \path [name intersections={of=lineFA1 and lineGB1,by=G2}];
             \node [fill=black,inner sep=1.5pt] at (G2) {};

         \draw[line width=.7mm] (A2) -- (B2) -- (C2) -- (D2) -- (E2) -- (F2) -- (G2) -- (A2);

        \draw[name path=lineAC2,line width=.2mm, dotted] (A2) -- (C2); 
        \draw[name path=lineCE2,line width=.2mm, dotted] (C2) -- (E2);
        \draw[name path=lineEG2,line width=.2mm, dotted] (E2) -- (G2);
        \draw[name path=lineGB2,line width=.2mm, dotted] (G2) -- (B2);
        \draw[name path=lineBD2,line width=.2mm, dotted] (B2) -- (D2);
        \draw[name path=lineDF2,line width=.2mm, dotted] (D2) -- (F2);
        \draw[name path=lineFA2,line width=.2mm, dotted] (F2) -- (A2);

         \path [name intersections={of=lineGB2 and lineAC2,by=A3}];
             \node [fill=black,inner sep=1.5pt] at (A3) {};
         \path [name intersections={of=lineAC2 and lineBD2,by=B3}];
             \node [fill=black,inner sep=1.5pt] at (B3) {};
         \path [name intersections={of=lineBD2 and lineCE2,by=C3}];
             \node [fill=black,inner sep=1.5pt] at (C3) {};
         \path [name intersections={of=lineCE2 and lineDF2,by=D3}];
             \node [fill=black,inner sep=1.5pt] at (D3) {};
         \path [name intersections={of=lineDF2 and lineEG2,by=E3}];
             \node [fill=black,inner sep=1.5pt] at (E3) {};
         \path [name intersections={of=lineEG2 and lineFA2,by=F3}];
             \node [fill=black,inner sep=1.5pt] at (F3) {};
         \path [name intersections={of=lineFA2 and lineGB2,by=G3}];
             \node [fill=black,inner sep=1.5pt] at (G3) {};

         \draw[line width=.6mm] (A3) -- (B3) -- (C3) -- (D3) -- (E3) -- (F3) -- (G3) -- (A3);

        \draw[name path=lineAC3,line width=.1mm, dotted] (A3) -- (C3); 
        \draw[name path=lineCE3,line width=.1mm, dotted] (C3) -- (E3);
        \draw[name path=lineEG3,line width=.1mm, dotted] (E3) -- (G3);
        \draw[name path=lineGB3,line width=.1mm, dotted] (G3) -- (B3);
        \draw[name path=lineBD3,line width=.1mm, dotted] (B3) -- (D3);
        \draw[name path=lineDF3,line width=.1mm, dotted] (D3) -- (F3);
        \draw[name path=lineFA3,line width=.1mm, dotted] (F3) -- (A3);

         \path [name intersections={of=lineGB3 and lineAC3,by=A4}];
             \node [fill=black,inner sep=1.5pt] at (A4) {};
         \path [name intersections={of=lineAC3 and lineBD3,by=B4}];
             \node [fill=black,inner sep=1.5pt] at (B4) {};
         \path [name intersections={of=lineBD3 and lineCE3,by=C4}];
             \node [fill=black,inner sep=1.5pt] at (C4) {};
         \path [name intersections={of=lineCE3 and lineDF3,by=D4}];
             \node [fill=black,inner sep=1.5pt] at (D4) {};
         \path [name intersections={of=lineDF3 and lineEG3,by=E4}];
             \node [fill=black,inner sep=1.5pt] at (E4) {};
         \path [name intersections={of=lineEG3 and lineFA3,by=F4}];
             \node [fill=black,inner sep=1.5pt] at (F4) {};
         \path [name intersections={of=lineFA3 and lineGB3,by=G4}];
             \node [fill=black,inner sep=1.5pt] at (G4) {};

         \draw[line width=.5mm] (A4) -- (B4) -- (C4) -- (D4) -- (E4) -- (F4) -- (G4) -- (A4);

        \draw[name path=lineAC4,line width=.05mm, dotted] (A4) -- (C4); 
        \draw[name path=lineCE4,line width=.05mm, dotted] (C4) -- (E4);
        \draw[name path=lineEG4,line width=.05mm, dotted] (E4) -- (G4);
        \draw[name path=lineGB4,line width=.05mm, dotted] (G4) -- (B4);
        \draw[name path=lineBD4,line width=.05mm, dotted] (B4) -- (D4);
        \draw[name path=lineDF4,line width=.05mm, dotted] (D4) -- (F4);
        \draw[name path=lineFA4,line width=.05mm, dotted] (F4) -- (A4);

         \path [name intersections={of=lineGB4 and lineAC4,by=A5}];
             \node [fill=black,inner sep=1pt] at (A5) {};
         \path [name intersections={of=lineAC4 and lineBD4,by=B5}];
             \node [fill=black,inner sep=1pt] at (B5) {};
         \path [name intersections={of=lineBD4 and lineCE4,by=C5}];
             \node [fill=black,inner sep=1pt] at (C5) {};
         \path [name intersections={of=lineCE4 and lineDF4,by=D5}];
             \node [fill=black,inner sep=1pt] at (D5) {};
         \path [name intersections={of=lineDF4 and lineEG4,by=E5}];
             \node [fill=black,inner sep=1pt] at (E5) {};
         \path [name intersections={of=lineEG4 and lineFA4,by=F5}];
             \node [fill=black,inner sep=1pt] at (F5) {};
         \path [name intersections={of=lineFA4 and lineGB4,by=G5}];
             \node [fill=black,inner sep=1pt] at (G5) {};

         \draw[line width=.4mm] (A5) -- (B5) -- (C5) -- (D5) -- (E5) -- (F5) -- (G5) -- (A5);

        \draw[name path=lineAC5,line width=.05mm, dotted] (A5) -- (C5); 
        \draw[name path=lineCE5,line width=.05mm, dotted] (C5) -- (E5);
        \draw[name path=lineEG5,line width=.05mm, dotted] (E5) -- (G5);
        \draw[name path=lineGB5,line width=.05mm, dotted] (G5) -- (B5);
        \draw[name path=lineBD5,line width=.05mm, dotted] (B5) -- (D5);
        \draw[name path=lineDF5,line width=.05mm, dotted] (D5) -- (F5);
        \draw[name path=lineFA5,line width=.05mm, dotted] (F5) -- (A5);

         \path [name intersections={of=lineGB5 and lineAC5,by=A6}];
             \node [fill=black,inner sep=0.7pt] at (A6) {};
         \path [name intersections={of=lineAC5 and lineBD5,by=B6}];
             \node [fill=black,inner sep=0.7pt] at (B6) {};
         \path [name intersections={of=lineBD5 and lineCE5,by=C6}];
             \node [fill=black,inner sep=0.7pt] at (C6) {};
         \path [name intersections={of=lineCE5 and lineDF5,by=D6}];
             \node [fill=black,inner sep=0.7pt] at (D6) {};
         \path [name intersections={of=lineDF5 and lineEG5,by=E6}];
             \node [fill=black,inner sep=0.7pt] at (E6) {};
         \path [name intersections={of=lineEG5 and lineFA5,by=F6}];
             \node [fill=black,inner sep=0.7pt] at (F6) {};
         \path [name intersections={of=lineFA5 and lineGB5,by=G6}];
             \node [fill=black,inner sep=0.7pt] at (G6) {};

         \draw[line width=.3mm] (A6) -- (B6) -- (C6) -- (D6) -- (E6) -- (F6) -- (G6) -- (A6);

        \draw[name path=lineAC6,line width=.06mm, dotted] (A6) -- (C6); 
        \draw[name path=lineCE6,line width=.06mm, dotted] (C6) -- (E6);
        \draw[name path=lineEG6,line width=.06mm, dotted] (E6) -- (G6);
        \draw[name path=lineGB6,line width=.06mm, dotted] (G6) -- (B6);
        \draw[name path=lineBD6,line width=.06mm, dotted] (B6) -- (D6);
        \draw[name path=lineDF6,line width=.06mm, dotted] (D6) -- (F6);
        \draw[name path=lineFA6,line width=.06mm, dotted] (F6) -- (A6);

         \path [name intersections={of=lineGB6 and lineAC6,by=A7}];
             \node [fill=black,inner sep=0.5pt] at (A7) {};
         \path [name intersections={of=lineAC6 and lineBD6,by=B7}];
             \node [fill=black,inner sep=0.5pt] at (B7) {};
         \path [name intersections={of=lineBD6 and lineCE6,by=C7}];
             \node [fill=black,inner sep=0.5pt] at (C7) {};
         \path [name intersections={of=lineCE6 and lineDF6,by=D7}];
             \node [fill=black,inner sep=0.5pt] at (D7) {};
         \path [name intersections={of=lineDF6 and lineEG6,by=E7}];
             \node [fill=black,inner sep=0.5pt] at (E7) {};
         \path [name intersections={of=lineEG6 and lineFA6,by=F7}];
             \node [fill=black,inner sep=0.5pt] at (F7) {};
         \path [name intersections={of=lineFA6 and lineGB6,by=G7}];
             \node [fill=black,inner sep=0.5pt] at (G7) {};

         \draw[line width=.2mm] (A7) -- (B7) -- (C7) -- (D7) -- (E7) -- (F7) -- (G7) -- (A7);

        \draw[name path=lineAC7,line width=.05mm, dotted] (A7) -- (C7); 
        \draw[name path=lineCE7,line width=.05mm, dotted] (C7) -- (E7);
        \draw[name path=lineEG7,line width=.05mm, dotted] (E7) -- (G7);
        \draw[name path=lineGB7,line width=.05mm, dotted] (G7) -- (B7);
        \draw[name path=lineBD7,line width=.05mm, dotted] (B7) -- (D7);
        \draw[name path=lineDF7,line width=.05mm, dotted] (D7) -- (F7);
        \draw[name path=lineFA7,line width=.05mm, dotted] (F7) -- (A7);

         \path [name intersections={of=lineGB7 and lineAC7,by=A8}];
         \path [name intersections={of=lineAC7 and lineBD7,by=B8}];
         \path [name intersections={of=lineBD7 and lineCE7,by=C8}];
         \path [name intersections={of=lineCE7 and lineDF7,by=D8}];
         \path [name intersections={of=lineDF7 and lineEG7,by=E8}];
         \path [name intersections={of=lineEG7 and lineFA7,by=F8}];
         \path [name intersections={of=lineFA7 and lineGB7,by=G8}];

         \draw[line width=.1mm] (A8) -- (B8) -- (C8) -- (D8) -- (E8) -- (F8) -- (G8) -- (A8);

        \draw[name path=lineAC8,line width=.04mm, dotted] (A8) -- (C8); 
        \draw[name path=lineCE8,line width=.04mm, dotted] (C8) -- (E8);
        \draw[name path=lineEG8,line width=.04mm, dotted] (E8) -- (G8);
        \draw[name path=lineGB8,line width=.04mm, dotted] (G8) -- (B8);
        \draw[name path=lineBD8,line width=.04mm, dotted] (B8) -- (D8);
        \draw[name path=lineDF8,line width=.04mm, dotted] (D8) -- (F8);
        \draw[name path=lineFA8,line width=.04mm, dotted] (F8) -- (A8);

         \path [name intersections={of=lineGB8 and lineAC8,by=A9}];
         \path [name intersections={of=lineAC8 and lineBD8,by=B9}];
         \path [name intersections={of=lineBD8 and lineCE8,by=C9}];
         \path [name intersections={of=lineCE8 and lineDF8,by=D9}];
         \path [name intersections={of=lineDF8 and lineEG8,by=E9}];
         \path [name intersections={of=lineEG8 and lineFA8,by=F9}];
         \path [name intersections={of=lineFA8 and lineGB8,by=G9}];

         \draw[line width=.1mm] (A9) -- (B9) -- (C9) -- (D9) -- (E9) -- (F9) -- (G9) -- (A9);

        \draw[name path=lineAC9,line width=.03mm, dotted] (A9) -- (C9); 
        \draw[name path=lineCE9,line width=.03mm, dotted] (C9) -- (E9);
        \draw[name path=lineEG9,line width=.03mm, dotted] (E9) -- (G9);
        \draw[name path=lineGB9,line width=.03mm, dotted] (G9) -- (B9);
        \draw[name path=lineBD9,line width=.03mm, dotted] (B9) -- (D9);
        \draw[name path=lineDF9,line width=.03mm, dotted] (D9) -- (F9);
        \draw[name path=lineFA9,line width=.03mm, dotted] (F9) -- (A9);

         \path [name intersections={of=lineGB9 and lineAC9,by=A10}];
         \path [name intersections={of=lineAC9 and lineBD9,by=B10}];
         \path [name intersections={of=lineBD9 and lineCE9,by=C10}];
         \path [name intersections={of=lineCE9 and lineDF9,by=D10}];
         \path [name intersections={of=lineDF9 and lineEG9,by=E10}];
         \path [name intersections={of=lineEG9 and lineFA9,by=F10}];
         \path [name intersections={of=lineFA9 and lineGB9,by=G10}];

         \draw[line width=.08mm] (A10) -- (B10) -- (C10) -- (D10) -- (E10) -- (F10) -- (G10) -- (A10);

        \draw[name path=lineAC10,line width=.01mm, dotted] (A10) -- (C10); 
        \draw[name path=lineCE10,line width=.01mm, dotted] (C10) -- (E10);
        \draw[name path=lineEG10,line width=.01mm, dotted] (E10) -- (G10);
        \draw[name path=lineGB10,line width=.01mm, dotted] (G10) -- (B10);
        \draw[name path=lineBD10,line width=.01mm, dotted] (B10) -- (D10);
        \draw[name path=lineDF10,line width=.01mm, dotted] (D10) -- (F10);
        \draw[name path=lineFA10,line width=.01mm, dotted] (F10) -- (A10);

         \path [name intersections={of=lineGB10 and lineAC10,by=A11}];
         \path [name intersections={of=lineAC10 and lineBD10,by=B11}];
         \path [name intersections={of=lineBD10 and lineCE10,by=C11}];
         \path [name intersections={of=lineCE10 and lineDF10,by=D11}];
         \path [name intersections={of=lineDF10 and lineEG10,by=E11}];
         \path [name intersections={of=lineEG10 and lineFA10,by=F11}];
         \path [name intersections={of=lineFA10 and lineGB10,by=G11}];

         \draw[line width=.06mm] (A11) -- (B11) -- (C11) -- (D11) -- (E11) -- (F11) -- (G11) -- (A11);
         
        \draw[name path=lineAC11,line width=.01mm, dotted] (A11) -- (C11); 
        \draw[name path=lineCE11,line width=.01mm, dotted] (C11) -- (E11);
        \draw[name path=lineEG11,line width=.01mm, dotted] (E11) -- (G11);
        \draw[name path=lineGB11,line width=.01mm, dotted] (G11) -- (B11);
        \draw[name path=lineBD11,line width=.01mm, dotted] (B11) -- (D11);
        \draw[name path=lineDF11,line width=.01mm, dotted] (D11) -- (F11);
        \draw[name path=lineFA11,line width=.01mm, dotted] (F11) -- (A11);

         \path [name intersections={of=lineGB11 and lineAC11,by=A12}];
         \path [name intersections={of=lineAC11 and lineBD11,by=B12}];
         \path [name intersections={of=lineBD11 and lineCE11,by=C12}];
         \path [name intersections={of=lineCE11 and lineDF11,by=D12}];
         \path [name intersections={of=lineDF11 and lineEG11,by=E12}];
         \path [name intersections={of=lineEG11 and lineFA11,by=F12}];
         \path [name intersections={of=lineFA11 and lineGB11,by=G12}];

         \draw[line width=.05mm] (A11) -- (B11) -- (C11) -- (D11) -- (E11) -- (F11) -- (G11) -- (A11);

        \draw[name path=lineAC12,line width=.01mm, dotted] (A12) -- (C12); 
        \draw[name path=lineCE12,line width=.01mm, dotted] (C12) -- (E12);
        \draw[name path=lineEG12,line width=.01mm, dotted] (E12) -- (G12);
        \draw[name path=lineGB12,line width=.01mm, dotted] (G12) -- (B12);
        \draw[name path=lineBD12,line width=.01mm, dotted] (B12) -- (D12);
        \draw[name path=lineDF12,line width=.01mm, dotted] (D12) -- (F12);
        \draw[name path=lineFA12,line width=.01mm, dotted] (F12) -- (A12);

         \path [name intersections={of=lineGB12 and lineAC12,by=A13}];
         \path [name intersections={of=lineAC12 and lineBD12,by=B13}];
         \path [name intersections={of=lineBD12 and lineCE12,by=C13}];
         \path [name intersections={of=lineCE12 and lineDF12,by=D13}];
         \path [name intersections={of=lineDF12 and lineEG12,by=E13}];
         \path [name intersections={of=lineEG12 and lineFA12,by=F13}];
         \path [name intersections={of=lineFA12 and lineGB12,by=G13}];

         \draw[line width=.04mm] (A12) -- (B12) -- (C12) -- (D12) -- (E12) -- (F12) -- (G12) -- (A12);

        \draw[name path=lineAC13,line width=.01mm, dotted] (A13) -- (C13); 
        \draw[name path=lineCE13,line width=.01mm, dotted] (C13) -- (E13);
        \draw[name path=lineEG13,line width=.01mm, dotted] (E13) -- (G13);
        \draw[name path=lineGB13,line width=.01mm, dotted] (G13) -- (B13);
        \draw[name path=lineBD13,line width=.01mm, dotted] (B13) -- (D13);
        \draw[name path=lineDF13,line width=.01mm, dotted] (D13) -- (F13);
        \draw[name path=lineFA13,line width=.01mm, dotted] (F13) -- (A13);

         \path [name intersections={of=lineGB13 and lineAC13,by=A14}];
         \path [name intersections={of=lineAC13 and lineBD13,by=B14}];
         \path [name intersections={of=lineBD13 and lineCE13,by=C14}];
         \path [name intersections={of=lineCE13 and lineDF13,by=D14}];
         \path [name intersections={of=lineDF13 and lineEG13,by=E14}];
         \path [name intersections={of=lineEG13 and lineFA13,by=F14}];
         \path [name intersections={of=lineFA13 and lineGB13,by=G14}];

         \draw[line width=.03mm] (A13) -- (B13) -- (C13) -- (D13) -- (E13) -- (F13) -- (G13) -- (A13);

        \draw[name path=lineAC14,line width=.01mm, dotted] (A14) -- (C14); 
        \draw[name path=lineCE14,line width=.01mm, dotted] (C14) -- (E14);
        \draw[name path=lineEG14,line width=.01mm, dotted] (E14) -- (G14);
        \draw[name path=lineGB14,line width=.01mm, dotted] (G14) -- (B14);
        \draw[name path=lineBD14,line width=.01mm, dotted] (B14) -- (D14);
        \draw[name path=lineDF14,line width=.01mm, dotted] (D14) -- (F14);
        \draw[name path=lineFA14,line width=.01mm, dotted] (F14) -- (A14);

         \path [name intersections={of=lineGB14 and lineAC14,by=A15}];
         \path [name intersections={of=lineAC14 and lineBD14,by=B15}];
         \path [name intersections={of=lineBD14 and lineCE14,by=C15}];
         \path [name intersections={of=lineCE14 and lineDF14,by=D15}];
         \path [name intersections={of=lineDF14 and lineEG14,by=E15}];
         \path [name intersections={of=lineEG14 and lineFA14,by=F15}];
         \path [name intersections={of=lineFA14 and lineGB14,by=G15}];

         \draw[line width=.03mm] (A14) -- (B14) -- (C14) -- (D14) -- (E14) -- (F14) -- (G14) -- (A14);

        \draw[name path=lineAC15,line width=.01mm, dotted] (A15) -- (C15); 
        \draw[name path=lineCE15,line width=.01mm, dotted] (C15) -- (E15);
        \draw[name path=lineEG15,line width=.01mm, dotted] (E15) -- (G15);
        \draw[name path=lineGB15,line width=.01mm, dotted] (G15) -- (B15);
        \draw[name path=lineBD15,line width=.01mm, dotted] (B15) -- (D15);
        \draw[name path=lineDF15,line width=.01mm, dotted] (D15) -- (F15);
        \draw[name path=lineFA15,line width=.01mm, dotted] (F15) -- (A15);

         \path [name intersections={of=lineGB15 and lineAC15,by=A16}];
         \path [name intersections={of=lineAC15 and lineBD15,by=B16}];
         \path [name intersections={of=lineBD15 and lineCE15,by=C16}];
         \path [name intersections={of=lineCE15 and lineDF15,by=D16}];
         \path [name intersections={of=lineDF15 and lineEG15,by=E16}];
         \path [name intersections={of=lineEG15 and lineFA15,by=F16}];
         \path [name intersections={of=lineFA15 and lineGB15,by=G16}];

         \draw[line width=.03mm] (A15) -- (B15) -- (C15) -- (D15) -- (E15) -- (F15) -- (G15) -- (A15);

    \end{tikzpicture}
\end{center}
\caption{The orbit of a convex polygon under the pentagram map converges to a point.}\label{Fig2}
\end{figure}

Specifically, suppose we are given an $n$-gon $P$ in the projectivization $\P\V$ of a $3$-dimensional vector space $\V$. Lift the vertices of the polygon to vectors $V_i \in \V$, $i = 1, \dots, n$. Define an operator $G_P \colon \V \to \V$ by the formula
\begin{align}\label{eq:glick}
G_p(V) := n V - \sum_{i=1}^n \frac{V_{i-1}\wedge V \wedge V_{i+1}}{V_{i-1}\wedge V_i \wedge V_{i+1}}\,V_i,
\end{align}
where all indices are understood modulo $n$. Note that this operator does not change under a rescaling of $V_i$'s and hence depends only on the polygon $P$. What Glick proved is that the limit point of successive images of $P$ under the pentagram map is one of the eigenvectors of $G_P$ (equivalently, a fixed point of the associated projective mapping $\P\V \to \P\V$).\par


We believe that the significance of Glick's operator actually goes beyond the limit point. In particular, as was observed by Glick himself, the operator $G_P$ has a natural geometric meaning for both pentagons and hexagons. Namely, by Clebsch's theorem every pentagon is projectively equivalent to its pentagram map image, and it turns out that the corresponding projective transformation is given by $G_P - 3I$, where $I$ is the identity matrix. Indeed, consider e.g. the first vertex of the pentagon and its lift $V_1$. Then the above formula gives
$$
(G_P - 3\Id)(V_1) = V_1 - \frac{V_{2}\wedge V_1 \wedge V_{4}}{V_{2}\wedge V_3 \wedge V_{4}}\,V_3 - \frac{V_{3}\wedge V_1 \wedge V_{5}}{V_{3}\wedge V_4 \wedge V_{5}}\,V_4.
$$
Taking the wedge product of this expression with $V_2 \wedge V_4$ or $V_3 \wedge V_5$ we get zero. This means that $$(G_P - 3\Id)(V_1) \in \mathrm{span}(V_2, V_4) \cap \mathrm{span}(V_3, V_5),$$ so the corresponding point in the projective plane is the intersection of diagonals of the pentagon. Furthermore, since Glick's operator is invariant under cyclic permutations, the same holds for all vertices, meaning that the operator $G_P - 3\Id$ indeed takes a pentagon to its pentagram map image.\par

Likewise, the second iterate of the pentagram map on hexagons also leads to an equivalent hexagon, and the equivalence is again realized by $G_P - 3I$. Finally, notice that for quadrilaterals  $G_P - 2I$ is a constant map onto the intersection of diagonals. These observations make us believe that the operator $G_P$ is per se an important object in projective geometry, whose full significance is yet to be understood.
\par
In the present paper we show that Glick's operator $G_P$ can be interpreted as \textit{infinitesimal monodromy}. To define the latter, consider the space of \textit{twisted polygons}, that are polygons closed up to a projective transformation, known as the \textit{monodromy}. Any closed polygon can be viewed as a twisted one, with trivial monodromy. To define the infinitesimal monodromy we deform a closed polygon into a genuine twisted one. To construct such a deformation, we use what is known as the \textit{scaling symmetry}. The scaling symmetry is a $1$-parametric group of transformations of twisted polygons which commutes with the pentagram map. That symmetry was instrumental for the proof of complete integrability of the pentagram map~\cite{ovsienko2010pentagram}. \par
Applying the scaling symmetry to a given closed polygon $P$ we get a family $P_z$ of polygons depending on a real parameter $z$ and such that $P_1 = P$. Thus, the monodromy $M_z$ of $P_z$ is a projective transformation depending on $z$ which is the identity for $z = 1$. By definition, the infinitesimal monodromy of $P$ is the derivative $dM_z/dz$ at $z = 1$. This makes the infinitesimal monodromy an element of the Lie algebra of the projective group $\PGL(\P^2)$, i.e. a linear operator on $\R^3$ defined up to adding a scalar matrix. The following is our main result.

\begin{theorem}\label{thm1}
The infinitesimal monodromy of a closed polygon $P$ coincides with Glick's operator $G_P$, up to addition of a scalar matrix.
\end{theorem}

This result provides another perspective on the limit point. Namely, observe that for $z \approx 1$ the monodromy $M_z$ of the deformed polygon is given by $$M_z \approx I + (z-1)(G_P + \lambda I),$$ up to higher order terms. Thus, the eigenvectors of $G_P$, and in particular the limit point, coincide with limiting positions of eigenvectors of $M_z$ as $z \to 1$. At least one of the eigenvectors of $M_z$ has a geometric meaning. Namely, the deformed polygon $P(z)$ can be thought of as a spiral, and the center of that spiral must be an eigenvector of the monodromy. We believe that as $z \to 1$ that eigenvector converges to the limiting point of the pentagram map (and not to one of the two other eigenvectors). If this is true, then we have the following picture. The scaling symmetry turns a closed polygon into a spiral. As the scaling parameter $z$ goes to $1$, the spiral approaches the initial polygon, while its center approaches the limit point of the pentagram map, see Figure \ref{Fig:scaling}.
\begin{figure}[h]
\centering
 \includegraphics[scale = 0.35]{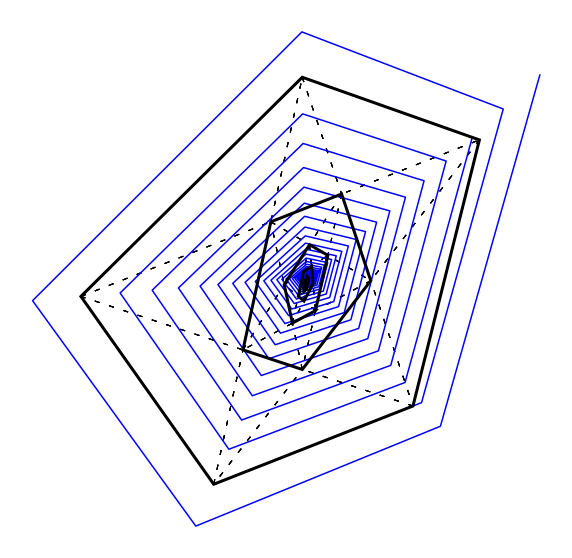}
    \caption{The image of a closed polygon under a scaling transformation is a spiral. As the scaling parameter goes to $1$, the center of the spiral approaches the limit point of the pentagram map. }\label{Fig:scaling}
    \end{figure}


 We note that the scaling symmetry is actually only defined on projective equivalence classes of polygons as opposed to actual polygons. This makes the family of polygons $P_z$ we used to define the infinitesimal monodromy non-unique. After reviewing basic notions in Section \ref{Sec:back}, we show in Section \ref{Sec:im} that the infinitesimal monodromy does not depend on the family used to define it. The proof of Theorem \ref{thm1} is given in Section \ref{Sec:proof}.
 
 We end the introduction by mentioning a possible future direction. The notion of infinitesimal monodromy is well-defined for polygons in any dimension and any scaling operation. For multidimensional polygons, there are different possible scalings, corresponding to different integrable generalizations of the pentagram map \cite{khesin2013, khesin2016}. It would be interesting to investigate the infinitesimal monodromy in those cases, along with its possible relation to the limit point of the corresponding pentagram maps. As for now, it is not even known if such a limit point exists for any class of  multidimensional polygons satisfying a convexity-type condition.
 \par It also seems that the infinitesimal monodromy in $\P^1$ is related to so-called cross-ratio dynamics, see \cite[Section 6.2.1]{arnold2018cross}.

\par
\bigskip

{\bf Acknowledgments.} The authors are grateful to Boris Khesin, Valentin Ovsienko, Richard Schwartz, and Sergei Tabachnikov for comments and discussions, as well to anonymous referees for their suggestions. A.I. was supported by NSF grant DMS-2008021.

\section{Background: twisted polygons, corner invariants, and scaling}\label{Sec:back}

In this section we briefly recall standard notions related to the pentagram map, concentrating on what will be used in the sequel.

A {\it twisted $n$-gon} is a bi-infinite sequence of points $v_i \in \P^2$ such that $v_{i+n}=M(v_i)$ for all $i\in \mathbb{Z}$ and a certain projective transformation $M\in \PGL(\P^2)$ called the {\it monodromy}. A twisted $n$-gon generalizes the notion of a closed $n$-gon as we recover a closed $n$-gon when the monodromy is equal to the identity. We denote the space of twisted $n$-gons by $\Poly_n$.

The pentagram map takes a twisted $n$-gon to a twisted $n$-gon (preserving the monodromy) so it can be regarded as a densely defined map from the space $\Poly_n$ of twisted $n$-gons to itself. From now on, we will assume that polygons are in sufficiently general position so as to allow for all constructions to go through unhindered. 

We say that two twisted $n$-gons $\{v_i\}$ and $\{v_i'\}$ are projectively equivalent when there is a projective transformation $\Phi$ such that $\Phi (v_i)=v_i'$. Notice, if two twisted $n$-gons are projectively equivalent, then their monodromies $M,M'$ are related by $M'=\Phi \circ M \circ \Phi^{-1}$. 




The pentagram map on twisted $n$-gons commutes with projective transformations
and as such descends to a map on the space $\Poly_n \,/\, \PGL(\P^2)$ of projective equivalence classes of twisted $n$-gons. 


We now recall a construction of coordinates on the space $\Poly_n \,/\, \PGL(\P^2)$ of projective equivalence classes of twisted $n$-gons. These coordinates are known as \textit{corner invariants} and were introduced in \cite{schwartz2008discrete}.



 Let $\{v_i \in \P^2\}$ be a twisted polygon. Then the corner invariants $x_i,y_i$ of the vertex $v_i$ are defined as follows.
\begin{align}
\begin{split}
    x_i &:=\Big[
        v_{i-2},v_{i-1},\big((v_{i-2},v_{i-1})\cap (v_i,v_{i+1})\big), \big((v_{i-2},v_{i-1})\cap (v_{i+1},v_{i+2})\big)
        \Big], \\
    y_i &:= \Big[
        \big( (v_{i-2},v_{i-1})\cap (v_{i+1},v_{i+2})\big), \big((v_{i-1},v_i)\cap (v_{i+1},v_{i+2}) \big), v_{i+1}, v_{i+2}
    \Big],
\end{split} \label{eq2}
\end{align}
where we define the cross-ratio $[a,b,c,d]$ of $4$ points $a,b,c,d$ on a projective line as
\begin{align}
  [a,b,c,d]:=\frac{(a-b)(c-d)}{(a-c)(b-d)}. \label{eq1}
\end{align}
Consider Figure \ref{Fig:CI}. The value of $x_i$ is the cross ratio of the four points drawn on the line $(v_{i-2},v_{i-1})$ (i.e. the line on the left) and $y_i$ is the cross ratio of the four points drawn on the line $(v_{i+1},v_{i+2})$ (i.e. the line on the right).

These corner invariants are defined on almost the entire space $\Poly_n$ of twisted $n$-gons.
Furthermore, these numbers are invariant under projective transformations and hence descend to the space $\Poly_n \,/\, \PGL(\P^2)$ of projective equivalence classes of twisted polygons.
As shown in \cite{schwartz2008discrete}, the functions $x_1, \dots, x_n, y_1, \dots, y_n$ constitute a coordinate system on an open dense subset of $\Poly_n \,/\, \PGL(\P^2)$. This in particular allows one to express the pentagram map, viewed as a transformation of $\Poly_n \,/\, \PGL(\P^2)$, in terms of the corner invariants.
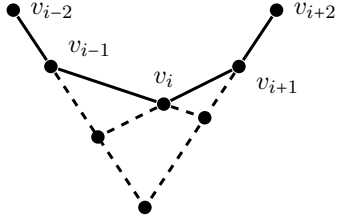
\begin{figure}[t]

\begin{center}
    \begin{tikzpicture}[thick,scale=0.5, mycirc/.style={circle,fill=black, inner sep=2pt},
    every node/.style={scale=0.9}]
        \node[mycirc, label=0:$v_{i-2}$] (v1) at (-4,2.5) {};
        \node[mycirc, label=10:$v_{i-1}$] (v2) at (-3,1) {};
        \node[mycirc, label=90:$v_{i}$] (v3) at (0,0) {};
        \node[mycirc, label=-10:$v_{i+1}$] (v4) at (2,1) {};
        \node[mycirc, label=0:$v_{i+2}$] (v5) at (3,2.5) {};
        
         \draw [name path=L12,line width=.4mm] (v1)--(v2); 
         \draw [name path=L54,line width=.4mm] (v5)--(v4); 
         \draw[name path=L43,line width=.4mm] (v4) -- (v3); 
         \draw[name path=L23,line width=.4mm] (v2) -- (v3); 
         
        \node[mycirc] (i1) at (-.5,-2.75) {};
        \node[mycirc] (i2) at (-1.75,-.875) {};
        \node[mycirc] (i3) at (1.0909,-.363636) {};
        
        \draw[name path=L2i1, line width=.4mm, dashed] (v2) -- (i1); 
        \draw[name path=L2i1, line width=.4mm, dashed] (v4) -- (i1);
        
        \draw[name path=L3i2, line width=.4mm, dashed] (v3) -- (i2); 
        \draw[name path=L3i3, line width=.4mm, dashed] (v3) -- (i3);
        
    \end{tikzpicture}

\end{center}
\caption{Definition of corner invariants.}\label{Fig:CI}
\end{figure}


If we are given a twisted $n$-gon with corner invariants $(x_i,y_i)$, then the corner invariants $(x_i,y_i)$ of its image under the pentagram are given by 
\begin{align}\label{eq3}
    x'_i=x_i\frac{1-x_{i-1} y_{i-1}}{1-x_{i+1} y_{i+1}} \qquad
    y'_i=y_{i+1} \frac{1-x_{i+2}y_{i+2}}{1-x_iy_i}. 
\end{align}
These formulas assume a specific labeling of vertices of the pentagram map image. For a different labeling the resulting formulas differ by a shift in indices. The choice of labeling, and more generally, the specific form of the above formulas will be of no importance to us. We will only use the following corollary.
Consider a $1$-parametric group of densely defined transformations $\Poly_n \,/\, \PGL(\P^2) \to \Poly_n \,/\, \PGL(\P^2)$ given by
\begin{align}\label{eq: scaling}
{R}_z\colon(x_i,y_i)\mapsto (x_i z, y_iz^{-1})
\end{align}
These transformations are known as \textit{scaling symmetries}.
\begin{proposition}
The scaling symmetry $R_z \colon \Poly_n \,/\, \PGL(\P^2) \to \Poly_n \,/\, \PGL(\P^2)$ on projective equivalence classes of twisted polygons commutes with the pentagram map for any $z \neq 0$.
\end{proposition}
\begin{proof}
The above formulas for the pentagram map in $x,y$ coordinates remain unchanged if all $x$ variables are multiplied by $z$ and all $y$ variables by are multiplied by $z^{-1}$.
\end{proof}
This proposition was a key tool in the proof of integrability of the pentagram map. Namely, consider a (twisted or closed) polygon $P$ defined up to a projective transformation, and let $P_z$ be its image under the scaling symmetry. Then, since the pentagram map commutes with scaling and preserves the monodromy, it follows that the monodromy $M_z$ of $P_z$ (which does not have to be the identity even if the initial polygon is closed!) is invariant under the map. Since $P_z$ is only defined as a projective equivalence class, this means that $M_z$ is only defined up to conjugation. Nevertheless, taking conjugation invariant functions (e.g. appropriately normalized eigenvalues) of $M_z$, we obtain, for every $z$, functions that are invariant under the pentagram map. It is shown in \cite{ovsienko2010pentagram} that the so-obtained functions commute under an appropriately defined Poisson bracket and turn the pentagram map into a discrete completely integrable system. See also \cite{ovsienko2013liouville} for a mode detailed proof. In our paper we utilize pretty much the same idea, but instead of looking at the eigenvalues of $M_z$ we will consider $M_z$ itself. It is not quite well-defined, but we will show that its $z$ derivative at $z = 1$ is, and that it coincides with Glick's operator.
\section{Infinitesimal monodromy}\label{Sec:im}
In this section we define the infinitesimal monodromy and show that it does not depend on the choices we need to make to formulate the definition, namely on the way we lift the scaling symmetry \eqref{eq: scaling} from projective equivalence classes of polygons to actual polygons.


We start with a closed $n$-gon, $P$, in $\P^2$. Let $[P] \in \Poly_n \,/\, \PGL(\P^2)$ be its projective equivalence class. Then, applying the scaling transformation $R_z$ given by \eqref{eq: scaling} to $[P]$, we get a path $R_z[P]$ in $ \Poly_n \,/\, \PGL(\P^2)$ such that $R_1[P] = [P]$. Now, choose a smooth in $z$ lift $P_z$ of the path $R_z[P]$ to the space $\Poly_n$ of actual twisted polygons such that $P_1 = P$ (we will construct an explicit example of such a lift later on). Denote by $M_z \in \PGL(\P^2)$ the monodromy of $P_z$. It is a family of projective transformations such that $M_1$ is the identity, $M_1 = \Id$. This family {does} depend on the choice of the lift $P_z$ of the path $R_z[P]$. However, as we show below, the tangent vector $dM_z/dz$ at $z = 1$ does not depend on that choice, and this is what we call the \textit{infinitesimal monodromy}.

\begin{definition}\label{def:im}
The \textit{infinitesimal monodromy} of a closed polygon $P$ is the derivative $dM_z / dz$ at $z=1$,
where $M_z$ is the monodromy of any path $P_z$ of polygons such that $P_1 = 1$ and $[P_z] = R_z[P]$.
\end{definition}
The infinitesimal monodromy is therefore a tangent vector to the projective group $\PGL(\P^2)$ at the identity, and, upon a choice of basis, can be viewed as a $3 \times 3$ matrix defined up to addition of a scalar matrix. Our main result can thus be formulated as follows.
\begin{theorem}[=Theorem \ref{thm1}]\label{thm2}
The tangent vector to $\PGL(\P^2)$ represented by Glick's operator $G_P$ coincides with the infinitesimal monodromy of $P$.
\end{theorem}
The proof will be given in Section \ref{Sec:proof}. But first we need to check that Definition \ref{def:im} makes sense, i.e. that the infinitesimal monodromy does not depend on the choice of the path $P_z$. This is established by the following:
\begin{proposition}
Let $P_z$ and $\tilde P_z$ be two families of polygons such that $P_1 = \tilde P_1$ is a closed polygon and $\tilde P_z$ is projectively equivalent to $P_z$ for every $z$. Then, for the monodromies $M_z$ and $\tilde M_z$ of these families, at $z = 1$ we have
$dM_z / dz = d\tilde M_z / dz$.
\end{proposition}
\begin{proof}
Let $\Phi_z$ be a projective transformation taking $P_z$ to $\tilde P_z$. Since $P_1 = \tilde P_1$, we have that $\Phi_1 = \Id$ (a generic $n$-gon in $\P^2$ does not admit any non-trivial projective automorphisms, provided that $n \geq 4$). Then we know that the monodromies are related by 
$\tilde M_z=\Phi_z M_z\Phi_z^{-1}.$
Differentiating this and using that $\Phi_1 = \Id$, we get
$$
\left.\frac{d}{dz}\right\vert_{z=1}\tilde M_z = 
\left.\frac{d}{dz}\right\vert_{z=1} M_z + \left[\left.\frac{d}{dz}\right\vert_{z=1} \Phi_z, M_1 \right].
$$
This identity in particular shows that the infinitesimal monodromy of a \textit{twisted} polygon is in general not well-defined, due to the extra commutator term in the right-hand side. But for a closed polygon we have $M_1 = \Id$, so the extra term vanishes and we get the desired identity.
\end{proof}

Before we proceed to the proof of the main theorem, let us mention one property of the infinitesimal monodromy: \begin{proposition}
The infinitesimal monodromy of a closed polygon is preserved by the pentagram map.
\end{proposition}
\begin{proof}
The pentagram map preserves the monodromy and commutes with the scaling. The infinitesimal monodromy is defined using monodromy and scaling and is thus preserved as well.
\end{proof}

This result in fact follows from our main theorem, because Glick shows in \cite[Theorem 3.1]{glick2020limit} that his operator has this property. However, the proof based on Glick's definition is quite non-trivial, while in our approach it is immediate. The observation that the infinitesimal monodromy is preserved by the pentagram map was in fact our motivation to conjecture that it should coincide with Glick's operator. And, as we show below, this is indeed true.

\section{The infinitesimal monodromy and Glick's operator}\label{Sec:proof}

In this section we prove our main result, Theorem \ref{thm1} (=Theorem \ref{thm2}). To that end, we explicitly construct a deformation $P_z$ of a polygon $P$ as in Definition~\ref{def:im}. Such a deformation is not unique, but we know that the infinitesimal monodromy does not depend on the deformation.  We will in fact use this ambiguity to our advantage by choosing a deformation for which the infinitesimal monodromy can be computed explicitly. We will then compute it and see that it coincides with Glick's operator.\par

Consider a closed $n$-gon $P$. Lift the $n$-periodic sequence $\{v_i \in \P^2\}$ of its vertices to an $n$-periodic sequence of non-zero vectors $V_i \in \R^3$. Then, for every $i \in \Z$, there exist $a_i, b_i, c_i \in \R$ such that
\begin{align}\label{eq4}
    V_{i+3}=a_i V_{i+2} + b_i V_{i+1} + c_i V_i. 
\end{align}
Furthermore, for a generic polygon the numbers $a_i, b_i, c_i$ are uniquely determined because the points $v_{i}, v_{i+1}, v_{i+2}$ are not collinear so the vectors $V_i, V_{i+1}, V_{i+2}$ are linearly independent. Also, we have $c_i \neq 0$ for any $i$ because the points  $v_{i+1}, v_{i+2}, v_{i+3}$ are not collinear.  In addition to that, since $V_{i+n} = V_i$ we have that the sequences $a_i,b_i,c_i$ are $n$-periodic. Finally, notice that for fixed $a_i, b_i, c_i$ the sequence $V_i$ is uniquely determined by equation \eqref{eq4} and initial condition $V_0, V_1, V_2$. Indeed, given $V_0, V_1, V_2$ and using that $c_i \neq 0$, we can successively find all $V_i$'s from  \eqref{eq4}. This gives us a way to deform the polygon $P$: keeping $V_0, V_1, V_2$ unchanged, we deform the coefficients in \eqref{eq4}. Namely, consider the following equation
\begin{align}\label{eq6}
     V_{i+3} = a_iV_{i+2} +z^{-1}\big( b_i V_{i+1}+c_iV_i \big), 
\end{align}
We assume that the vectors $V_0, V_1, V_2$ do not depend on $z$ and coincide with the above-constructed lifts of vertices of $P$. For any $z \neq 0$, equation \eqref{eq6} has a unique solution with such initial condition. For $z = 1$ we recover the initial polygon, while for other values of $z$ we get its deformation. Note that for $i \neq 0,1,2$ the solutions $V_i$ of \eqref{eq6} are actually functions of the parameter $z$, i.e. $V_i = V_i(z)$. 
\begin{proposition}
Taking the solution of \eqref{eq6} such that $V_0, V_1, V_2$ are fixed lifts of vertices $v_0, v_1, v_2$ of $P$ and projecting the vectors $V_i \in \R^3$ to $\P^2$, we get a family $P_z$ of twisted polygons as in Definition~\ref{def:im}. Namely, we have that $P_1 = P$, and also $[P_z] = R_z[P]$, where $R_z$ is the scaling symmetry \eqref{eq: scaling}.
\end{proposition}
\begin{proof}
First note that if a sequence $V_i$ is a solution of \eqref{eq6} with given initial condition, then $V_i(z) \neq 0$ for any $i$ and every $z$ sufficiently close to $1$, so we can indeed project those vectors to get a sequence of points in $\P^2$. Indeed, for $z = 1$ this is so by construction and hence is also true for nearby values of $z$ by continuity (in fact, one can show that $V_i(z) \neq 0$ for any $z \neq 0$, not necessarily close to $1$). 

Further, observe that since the coefficients of equation~\eqref{eq6} are periodic, its solution is quasi-periodic: $V_{i+n}(z) = M_z V_i(z)$ for a certain invertible matrix $M_z$ depending on $z$. Therefore, the projections $v_i(z) \in \P^2$ of the vectors $V_i(z) \in \R^3$ form a twisted polygon whose monodromy is the projective transformation defined by $M_z$. Furthermore since equations~\eqref{eq4} and \eqref{eq6} agree for $z = 1$, and the initial conditions are the same too, it follows that for the so-obtained family $P_z$ of twisted polygons we have $P_1 = P$. Finally, we need to show that the projective equivalence classes of $P$ and $P_z$ are related by scaling $[P_z] = R_z[P]$. To that end, we use formulas expressing corner invariants in terms of coefficients of a recurrence relation satisfied by the lifts of vertices.
Arguing as in the proof of \cite[Lemma 4.5]{ovsienko2010pentagram}
one gets the following expressions for the corner invariants of $P$:
\begin{align}
    x_{i+2}=\frac{a_ic_i}{b_i b_{i+1}} \qquad y_{i+2}=-\frac{b_{i+1} }{a_i a_{i+1}}. 
\end{align}
Accordingly, since equations \eqref{eq4} and \eqref{eq6} encoding $P$ and $P_z$ are connected by the transformation $b_i \mapsto z^{-1}b_i$, $c_i \mapsto z^{-1}c_i $, the corner invariants of $P_z$ are given by
\begin{align}
    x_{i+2}(z)=\frac{a_i(z^{-1}c_i)}{(z^{-1}b_i) (z^{-1}b_{i+1})} = z x_{i+2} \qquad y_{i+2}(z)=-\frac{z^{-1}b_{i+1} }{a_i a_{i+1}} = z^{-1}y_{i+2}.
\end{align}
Thus, the projective equivalence classes of the polygons $P$ and $P_z$ are indeed related by scaling, as desired.
\end{proof}

We are now in a position to prove our main result. To that end, we will compute the monodromy of the polygon defined by \eqref{eq6}, take its derivative at $z=1$, and hence find the infinitesimal monodromy. 

We put the vectors $V_i(z)$ into columns of matrices as follows: define $$W_i(z):=\big[V_{i+2}(z) \quad V_{i+1}(z) \quad V_i(z) \big].$$ Then the relation \eqref{eq6} gives us the matrix equation
$$W_{i+1}(z)=W_i(z) U_i(z),$$ where 
\begin{align}\label{eq:ui}
    U_i(z):=\begin{bmatrix}
a_i & 1 & 0 \\
z^{-1}b_i & 0 & 1 \\
z^{-1}c_i & 0 & 0
\end{bmatrix}.
\end{align}
We stop explicitly recording the dependence on $z$ as it is notationally cumbersome.
Inductively, we have that 
\begin{align}\label{eq11}
W_i = W_0U_0U_1 \dots U_{i-1}.
\end{align} 
In particular,
$$
   W_{n}=W_0U
$$
where $U:= U_0U_{1}\dots U_{n-1}$. 
At the same time, we have that $V_{i+ n}=M_zV_i$, where $M_z$ is a matrix representing the monodromy of the polygon defined by the vectors $V_i$.
This means that $W_{n}=M_zW_0$. Relating these two expressions for $W_{n}$ we get 
\[ W_0U=M_zW_0  \quad \iff \quad M_z=W_0UW_0^{-1}.\]
 Notice that because $V_0,V_1,V_2$ are fixed we have that $W_0=[V_0 \quad V_1 \quad V_2]$ is constant while $z$ varies. This means that all the dependence of $M_z$ on $z$ is contained in the expression for $U$. This gives
\begin{align*}
   \frac{dM_z}{dz} &=\frac{d}{dz} \left( W_0 U_0 \dots U_{n-1}  W_0^{-1} \right) 
   \\ &=  \sum_{i=0}^{n-1} W_0 U_0 \dots U_{i-1}\frac{dU_i}{dz}U_{i+1} \dots U_{n-1}   W_0^{-1} 
   = \sum_{i=0}^{n-1} W_i \frac{dU_i}{dz}U_{i+1} \dots U_{n-1}   W_0^{-1},
\end{align*}
where the last equality uses that $W_{i}=W_0U_0\dots U_{i-1}$. Further, observe that \begin{align}U_{i+1}\dots U_{n-1}&=(U_0\dots U_{i})^{-1}(U_0\dots U_{n-1}) =(W_0^{-1}W_{i+1})^{-1}(W_0^{-1}W_{n})=W_{i+1}^{-1}W_{n}.\end{align} 
Also using that $W_n W_0^{-1}=M_z$, we get
\begin{align*}
   \frac{dM_z}{dz} = \sum_{i=0}^{n-1} W_i \frac{dU_i}{dz} W_{i+1}^{-1}W_{n}  W_0^{-1} = \left(\sum_{i=0}^{n-1} W_i \frac{dU_i}{dz} W_{i+1}^{-1}\right)M_z.
\end{align*}
Further, using that the monodromy satisfies $M_1=\Id$ because we started with a closed $n$-gon, we arrive at
$$
\left.\frac{dM_z}{dz}\right|_{z=1} \!\!\!= \,\,\sum_{i=0}^{n-1} S_i,
$$
where
$$
S_i := \left.\left(W_i \frac{dU_i}{dz} W_{i+1}^{-1}\right)\,\right\vert_{z=1}.
$$
Now, we will show that summing these $S_i$ with $i=0,1,\ldots,n-1$ gives \eqref{eq:glick} up to a scalar matrix. Using \eqref{eq:ui}, we get
\begin{align*}
\left.\frac{dU_i}{dz}\right|_{z=1} \!\!=\,\, 
\begin{bmatrix}
0 & 0 & 0 \\
-b_i & 0 & 0 \\
-c_i & 0 & 0
\end{bmatrix}.
\end{align*}
Further, observe that for $z=1$ the matrix $W_i$ sends the standard basis to the lifts $V_{i+2}, V_{i+1}, V_i$ of the vertices of $P$. Therefore $W_{i+1}^{-1}$ takes the vectors $V_{i+3},V_{i+2},V_{i+1}$ to the standard basis, from which we find that the matrix $S_i$ acts on these vectors as
$$
 V_{i+3} \mapsto -b_iV_{i+1} - c_iV_i,
       \quad V_{i+2} \mapsto 0
       \quad V_{i+1} \mapsto 0.
$$
Using also \eqref{eq4}, we find that
$$
S_i(V_i) = \frac{1}{c_i}S_i(V_{i+3}) =  -\frac{b_i}{c_i} V_{i+1} - V_i,
$$
which means that 
\begin{align}
    S_i(V)=\frac{|V,V_{i+1},V_{i+2}|}{|V_i,V_{i+1},V_{i+2}|}\Big(-V_i-\frac{b_i}{c_i}V_{i+1}\Big) \quad \forall \,\,V \in \R^3, \label{eq13}
\end{align}
where $|A,B,C|$ is the determinant of the matrix with columns $A, B, C$. Further, rewriting \eqref{eq4} as
$$
-V_i-\frac{b_i}{c_i}V_{i+1} = \frac{a_i}{c_i}V_{i+2} -\frac{1}{c_i}V_{i+3}
$$
we get
\begin{align*}
    S_i(V)=
    \frac{|V_{i+1},V_{i+2},V|}{|V_{i+1},V_{i+2},V_i|}\Big( \frac{a_i}{c_i}V_{i+2} -\frac{1}{c_i}V_{i+3} \Big) = \frac{|V_{i+1},V_{i+2},V|}{|V_{i+1},V_{i+2},c_i^{-1}V_{i+3}|}\Big( \frac{a_i}{c_i}V_{i+2} -\frac{1}{c_i}V_{i+3} \Big),
\end{align*}
where in the last equality we used \eqref{eq4} to express $V_i$ in terms of $V_{i+1}, V_{i+2}, V_{i+3}$. 
This can be rewritten as
\begin{align}\label{eq:si}
\begin{aligned}
    S_i(V)= \frac{|V_{i+1},V_{i+2},V|}{|V_{i+1},V_{i+2},V_{i+3}|} a_iV_{i+2} - \frac{|V_{i+1},V_{i+2},V|}{|V_{i+1},V_{i+2},V_{i+3}|} V_{i+3},
\end{aligned}
\end{align}
and the first term can be further rewritten as
\begin{align}\label{eq:si2}
\begin{aligned}
    \frac{|V_{i+1},V_{i+2},V|}{|V_{i+1},V_{i+2},V_{i+3}|} a_iV_{i+2} &= 
    \frac{|V_{i+1},a_iV_{i+2},V|}{|V_{i+1},V_{i+2},V_{i+3}|} V_{i+2} 
   = \frac{|V_{i+1},V_{i+3}-c_iV_i,V|}{|V_{i+1},V_{i+2},V_{i+3}|} V_{i+2} \\ &= - \frac{|V_{i+1},V,V_{i+3}|}{|V_{i+1},V_{i+2},V_{i+3}|} V_{i+2} + \frac{|V_{i},V_{i+1},V|}{|V_{i+1},V_{i+2},V_{i+3}|} c_iV_{i+2}
   \end{aligned}
\end{align}
where in the second equality we used \eqref{eq4} to express $a_iV_{i+2}$ in terms of $V_i, V_{i+1}, V_{i+3}$. Furthermore, using \eqref{eq4} to express $V_{i+3}$ in terms of $V_i, V_{i+1}, V_{i+2}$, the last term in the latter expression can be rewritten as
\begin{align}\label{eq:si3}
\frac{|V_{i},V_{i+1},V|}{|V_{i+1},V_{i+2},V_{i+3}|} c_iV_{i+2} = \frac{|V_{i},V_{i+1},V|}{|V_i, V_{i+1},V_{i+2}|} V_{i+2}.
\end{align}
Combining \eqref{eq:si}, \eqref{eq:si2}, and \eqref{eq:si3}, we arrive at the following expression
\begin{align}
    S_i(V)&=
  -\frac{|V_{i+1},V,V_{i+3}|}{|V_{i+1},V_{i+2},V_{i+3}|} V_{i+2} + \frac{|V_{i},V_{i+1},V|}{|V_i, V_{i+1},V_{i+2}|} V_{i+2}
    - \frac{|V_{i+1},V_{i+2},V|}{|V_{i+1},V_{i+2},V_{i+3}|} V_{i+3}. \label{eq15}
\end{align}
Since the last two terms only differ by a shift in index, and the sequence of $V_i$'s in $n$-periodic, we get 
$$
\left.\frac{dM_z}{dz}\right|_{z=1} \!\!\!(V)= -\sum_{i=0}^{n-1} S_i(V) = \sum_{i=0}^{n-1} \frac{|V_{i+1},V,V_{i+3}|}{|V_{i+1},V_{i+2},V_{i+3}|} V_{i+2} = -\sum_{i=0}^{n-1} \frac{|V_{i-1},V,V_{i+1}|}{|V_{i-1},V_{i},V_{i+1}|} V_{i}, 
$$
which coincides with Glick's operator \eqref{eq:glick} up to a scalar matrix. Thus, Theorem \ref{thm1} (=Theorem~\ref{thm2}) is proved.

\bibliographystyle{plain}
\bibliography{ref}

\end{document}